\documentclass[conference,letterpaper]{IEEEtran}
\IEEEoverridecommandlockouts
\usepackage[utf8]{inputenc}
\usepackage[T1]{fontenc}
\usepackage{url}
\usepackage{ifthen}
\usepackage{cite}
\usepackage[cmex10]{amsmath} 

\usepackage{geometry}
\usepackage{amsmath}
\usepackage{graphicx}
\usepackage{subfigure}
\usepackage{float}
\usepackage{makecell}
\usepackage{enumerate}
\usepackage{multirow,tabularx}
\usepackage{caption}
\usepackage{epstopdf}
\usepackage{amsthm,amssymb,mathrsfs}
\theoremstyle{plain}
\newtheorem{theorem}{Theorem}

\theoremstyle{definition}

\newtheorem{remark}{Remark}

\makeatletter
\renewcommand{\maketag@@@}[1]{\hbox{\m@th\normalsize\normalfont#1}}%
\makeatother

\begin{document}

\title{Optimal Feedback Schemes for Dirty Paper Channels With State Estimation at the Receiver}
\author{\IEEEauthorblockN{
Dengfeng~Xia\IEEEauthorrefmark{1}\IEEEauthorrefmark{2},
 Han~Deng\IEEEauthorrefmark{1},
 Haonan~Zhang\IEEEauthorrefmark{1},
 Fan~Cheng\IEEEauthorrefmark{4},
Bin~Dai\IEEEauthorrefmark{1}\IEEEauthorrefmark{2}\IEEEauthorrefmark{3}
Liuguo~Yin\IEEEauthorrefmark{5}}
\IEEEauthorblockA{
\IEEEauthorrefmark{1}
School of Information Science and Technology, Southwest Jiaotong University, Chengdu, 610031, China.}
\IEEEauthorblockA{
\IEEEauthorrefmark{2}
Chongqing Key Laboratory of Mobile Communications Technology, Chongqing 400065, China.}
\IEEEauthorblockA{
\IEEEauthorrefmark{3}
Provincial
Key Lab of Information Coding and Transmission, Southwest Jiaotong
University, Chengdu, 611756, China.}
\IEEEauthorblockA{
\IEEEauthorrefmark{4}
Department of Computer Science and Engineering, Shanghai Jiao Tong University, Shanghai, 200241, China.}
\IEEEauthorblockA{
\IEEEauthorrefmark{5}
Beijing National Research Center for Information Science and Technology, Tsinghua University, Beijing 100084, China.\\
\{xiadengf,denghan,zhanghaonan\}@my.swjtu.edu.cn, chengfan@cs.sjtu.edu.cn, daibin@home.swjtu.edu.cn, yinlg@tsinghua.edu.cn.}}
\maketitle

\begin{abstract}
In the literature, it has been shown that feedback does not increase the optimal rate-distortion region of
the dirty paper channel with state estimation at the receiver (SE-R). On the other hand, it is well-known that feedback helps to construct  low-complexity coding schemes in Gaussian channels, such as the elegant Schalkwijk-Kailath (SK) feedback scheme. This motivates us to explore capacity-achieving SK-type schemes in dirty paper channels with SE-R and feedback.
In this paper, we first propose a capacity-achieving feedback scheme for the dirty paper channel with SE-R (DPC-SE-R), which combines the superposition coding and the classical SK-type scheme.
Then, we extend this scheme to the dirty paper multiple-access channel with SE-R and feedback, and also show the extended scheme is capacity-achieving.
Finally, we discuss how to extend our scheme to a noisy state observation case of the DPC-SE-R.
However, the capacity-achieving SK-type scheme for such a case remains unknown.

\end{abstract}


\section{Introduction}\label{sec1}

The dirty paper channel (DPC) \cite{co}, which is the additive white Gaussian noise (AWGN) channel with additive Gaussian state interference non-causally known at the transmitter, receives much attention in the literature.
The problem of joint state estimation and communication (JSEC) over
the DPC
was first investigated by \cite{jsec-first}, where the receiver simultaneously decodes the transmitted message and estimates the state interference, and the optimal trade-off between the transmission rate and the mean squared error distortion in estimating the state was totally determined in \cite{jsec-first}.
Recently, \cite{jsec-fb} further investigated the impact of feedback in the DPC with state estimation at the receiver (SE-R)\footnote{In \cite{jsec-fb}, it was shown that the problem of JSEC over the DPC with noiseless feedback is the source-coding problem with a vending machine\cite{app1}, which provides a useful model for computer networks that acquire data from remote databases and sensor networks that acquire data through measurements\cite{app2}.}, and showed that
feedback does not enhance the optimal rate-distortion  trade-off region of the DPC with SE-R.

On the other hand, it is well-known that feedback helps to construct low-complexity coding schemes, such as the elegant Schalkwijk-Kailath (SK) feedback scheme \cite{sk} for the AWGN channel. It has been shown that the SK scheme is not only capacity-achieving, but also has extremely low encoding-decoding complexity since a simple first-order recursive coding approach is applied to the transceiver.
Recently, \cite{dpc-fed} and \cite{nm-tw} extended the SK scheme to the DPC with feedback, and showed that modified SK-type schemes perfectly eliminate non-causal channel state interference and achieve the channel capacity. The basic intuition behind these schemes is that since the transmitter knows all state interference in advance, an estimation offset caused by these state interference in the SK scheme can be computed by the transmitter. Then at the first time instant, if the transmitter inserts a negative value of this estimation offset into the transmission codeword, the entire offset can be perfectly eliminated when the transmission is completed.

Motivated by the excellent performance of the SK-type schemes in the DPC, in this paper, we aim to explore how to design capacity-achieving SK-type schemes for dirty paper channels with SE-R and feedback, and the contribution of this paper is summarized as follows.

First, we propose a capacity-achieving SK-type scheme for the DPC with SE-R and feedback, which
combines the superposition coding for the broadcast channel and the SK-type scheme for the DPC.
Next, we further extend the above scheme to a  multiple-access case, namely,
the dirty paper multiple-access channel (DP-MAC) with SE-R and feedback.
By establishing a corresponding upper bound, we show that this extended scheme is optimal, and determine the optimal rate-distortion trade-off region of the DP-MAC with SE-R and feedback, which remains unknown in the literature.
Finally, we show that by using an equivalent channel model, our proposed schemes can be directly extended to a noisy state observation case of the DPC with SE-R and feedback. However, we should point out that the capacity-achieving SK-type scheme for such a noisy observation case remains unknown.

The remainder of this paper is organized as follows.
An optimal feedback scheme for the DPC with SE-R is provided in Section \ref{sec-2}.
Section \ref{sec-3} shows the main results about the DP-MAC with SE-R and feedback.
Section \ref{sec-4} discusses a noisy state observation case of the DPC with SE-R and feedback.

\section{The DPC with SE-R and Feedback}\label{sec-2}
\subsection{Model Formulation}\label{sec2-1}
\begin{figure}[h!]
	\centering
	\includegraphics[width=0.96\linewidth]{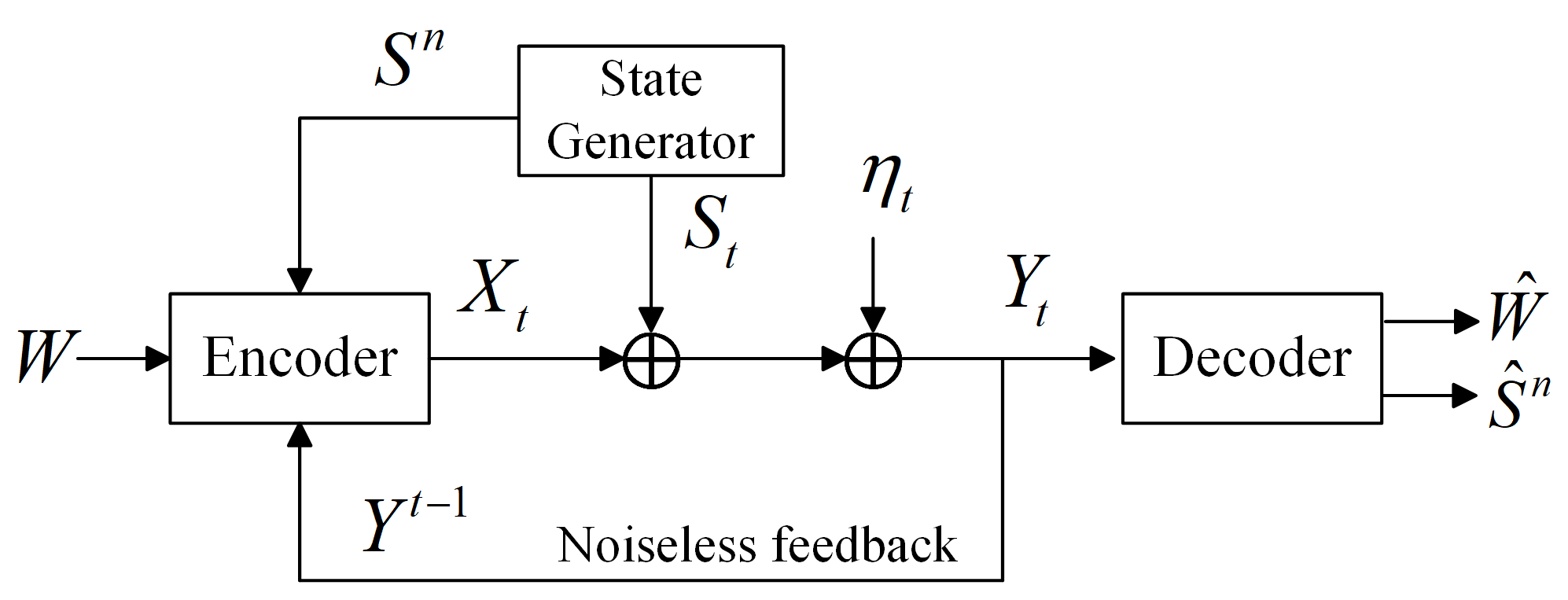}
	\caption{The DPC with SE-R and  feedback}
	\label{figmod}
\end{figure}
The DPC with SE-R and feedback is depicted in Fig. \ref{figmod}, and
the channel output at time instant $t$ is
\begin{equation}
	Y_t=X_t+S_t+\eta_t, \, t=1,2,\ldots,n,
\end{equation}
where $X_t$ is the channel input, $\eta_t$ is an independent identically distributed (i.i.d.) Gaussian noise with zero-mean and  variance $\sigma^2$, i.e., $\eta_t \sim \mathcal{N}(0,\sigma^2)$,
and
$S_t$ is an i.i.d. Gaussian state interference with zero-mean and variance $Q$, i.e., $S_t \sim \mathcal{N}(0,Q)$.
The state sequence $S^n=(S_1,S_2,\ldots,S_n)$ is known by the transmitter in a non-causal manner. In addition, we assume that $\eta_t$, $S_t$ and the message $W$ are independent of each other.

An $(n,R,P)$-code for the model of Fig. \ref{figmod}  consists of:
\begin{enumerate}[1)]
	\item a uniformly distributed message $W$ takes values in the set $\mathcal{W}= \{1,2,\ldots,2^{nR}\}$;
	\item an  encoder with output $X_t=\varphi_t(W,S^n,Y^{t-1})$
	satisfying the average power constraint $\frac{1}{n}\sum\nolimits_{i=1}^{n}E[X^2_{t}]\leq P$, where $\varphi_{t}$ is the $t$-th encoding function;
	\item  a decoder with outputs $\hat{W}=\psi(Y^n)$ and $\hat{S}^n=\phi(Y^n)$, where $\psi$ and $\phi$ are the decoding and estimation functions, respectively.
\end{enumerate}
The average decoding error probability is defined as
\begin{equation}\label{pe-def}
	P_{e}^{(n)}=\frac{1}{2^{nR}}\sum_{w=1}^{2^{nR}}\Pr\{\psi(Y^n)\neq w|w\;\mbox{was sent}\},
\end{equation}
and the mean-squared state estimation error is  defined as
\begin{equation}\label{ee-def}
	{\rm E}d(S^n,\hat{S}^n)=\frac{1}{n}\sum_{t=1}^{n}{\rm E}[(S_t-\hat{S}_t)^2].
\end{equation}

A rate-distortion (R-D) pair $(R, D)$ is \emph{achievable} if for  any $\epsilon>0$, there exists the above  $(n,R,P)$-code such that $\frac{1}{n}H(W)\ge R-\epsilon$, $P_e^{(n)}\le \epsilon$, and ${\rm E}d(S^n,\hat{S}^n)\le D+\epsilon$ as $n\to \infty$.
The optimal R-D trade-off region  for the DPC with SE-R and feedback is the closure of the convex hull of all achievable $(R, D)$ pairs, denoted as $\mathscr{C}_{\rm dp}^{\rm  fb}$.

The optimal R-D trade-off region $\mathscr{C}_{\rm dp}^{\rm  fb}$ was characterized in \cite[Remark 3]{jsec-fb},
which is given by
\begin{equation}\label{c1}
	\begin{aligned}
		\mathscr{C}_{\rm dp}^{\rm fb}=\bigcup_{0 \le \gamma \le 1}
		\Big
		\{(R,D):
		0\le R\le \frac{1}{2}\log \Big(1+\frac{\gamma P}{\sigma^2}\Big),&\\
		D\ge Q\frac{(\gamma P+\sigma^2)}{(\sqrt{Q}+\sqrt{(1-\gamma)P})^2+\gamma P+\sigma^2}\Big\}&.
	\end{aligned}
\end{equation}

\subsection{An Optimal Feedback Scheme for the DPC with SE-R}\label{sec2-2}
In this subsection, we propose an SK-type feedback scheme for the DPC with SE-R that achieves $\mathscr{C}_{\rm dp}^{\rm  fb}$.

The main idea of this SK-type feedback scheme is explained below.
The transmitted codeword is divided into two parts, i.e.,
\begin{equation}
	X^n=V^n+G^n,
\end{equation}
where the power of $V^n$ is $(1-\gamma)P$, the power of $G^n$ is $\gamma P$, and $0\le \gamma \le 1$ is a power allocation coefficient.
When the transmitter sends $X^n=V^n+G^n$, the received channel output is
\begin{equation}
		Y^n=V^n+G^n+S^n+\eta^n.
\end{equation}
First, since the transmitter knows the channel state $S^n$ in a non-causal manner, the codeword $V^n$ is used to directly transmit $S^n$, i.e., $V^n=\sqrt{\frac{(1-\gamma)P}{Q}}S^n$, which helps to improve the state estimation performance of the receiver.
Then, since $S^n$ and $V^n$ are known by the transmitter, the codeword $G^n$ is used to transmit the message $W$ by applying the SK-type scheme for the DPC with a new non-causal state $V^n+S^n$.
The detailed scheme is given below.

\textit{Encoding:}
At time $1$, the transmitter sends
\begin{equation}\label{x11}
	X_1=\underbrace{\sqrt{12\gamma P}(\theta-O)}_{G_1} + \underbrace{\sqrt{\frac{(1-\gamma)P}{Q}}S_1}_{V_1},
\end{equation}
where $\theta=-\frac{1}{2}+\frac{(2W-1)}{2\cdot2^{nR}}$  is a mapped value of the message $W$ that is the same as the classical SK scheme \cite{sk},
and the offset $O$ is a linear function about state interference that is given by
\begin{equation}
\begin{aligned}
	O=\frac{1}{\sqrt{12\gamma P}}\omega S_1
	-\omega \sum_{i=2}^{n}\mu_{i}S_i,
\end{aligned}	
\end{equation}
$\omega \triangleq 1+ \sqrt{\frac{(1-\gamma)P}{Q}}$
and $\mu_{i}$ will be defined later.
The receiver gets
\begin{equation}
	Y_1=X_1+S_1+\eta_1
=\sqrt{12\gamma P}(\theta-O) + \omega S_1 +\eta_1,
\end{equation}
and calculates the first estimation of $\theta$ as
\begin{equation}\label{es1}
\hat{\theta}_1=\frac{Y_1}{\sqrt{12\gamma P}}
=\theta+\omega\sum_{i=2}^{n}\mu_{i}S_i+\frac{\eta_1}{\sqrt{12\gamma P}}.
\end{equation}
Define $\varepsilon_1 \triangleq \frac{\eta_1}{\sqrt{12\gamma P}}$, and let $\alpha_1\triangleq {\rm Var}[\varepsilon_1]=\frac{\sigma^2}{12\gamma P}$.

At time $t$ ($t\in\{2,3,\ldots,n\}$), the transmitter receives the feedback signal $Y_{t-1}=X_{t-1}+S_{t-1}+\eta_{t-1}$, and
encodes
\begin{equation}\label{xte}
	X_t=\underbrace{\sqrt{\frac{\gamma P}{\alpha_{t-1}}}\varepsilon_{t-1}}_{G_t}+\underbrace{\sqrt{\frac{(1-\gamma)P}{Q}}S_t}_{V_t},
\end{equation}
where $\alpha_{t-1}\triangleq {\rm Var}[\varepsilon_{t-1}]$,
\begin{align}
	\varepsilon_{t-1}&=\varepsilon_{t-2}-\mu_{t-1}(Y_{t-1}-\omega S_{t-1}),\label{eded}\\
		\mu_{t-1}&=\frac{{\rm E}[\varepsilon_{t-2}(Y_{t-1}-\omega S_{t-1})]}
		{{\rm E}[(Y_{t-1}-\omega S_{t-1})^2]}.\label{mude}
\end{align}

From (\ref{xte})-(\ref{mude}), we calculate that
\begin{equation}
	\mu_{i}=\frac{\sqrt{\gamma P \alpha_{t-1}}}{\gamma P+\sigma^2}, \;\; \alpha_t=\alpha_{t-1}\frac{\gamma P}{\gamma P+\sigma^2}.
\end{equation}
Note that from (\ref{x11}), we conclude that the power of $X_1$ is $\gamma P+\left(1+12\gamma P\omega^2\sum_{i=2}^{n}\mu_{i}^2\right)Q$ and is bounded
since $\{\mu_i\}$ is a geometric sequence where the common ratio is $\sqrt{\frac{\gamma P}{\gamma P+\sigma^2}}$.
From (\ref{xte}), the power of $X_t$ ($t\ge 2$) is $P$ since $\epsilon_{t-1}$ is a linear function about $(\eta_1,\eta_2,\ldots\eta_{t-1})$, and it is independent of $S_t$.
Hence, as $n \to \infty$, the average transmission power tends to $P$, which is similar to the arguments in \cite{nm-tw}.

\textit{Decoding:}
At time $t$ ($t\in\{2,3,\ldots,n\}$), the $t$-th estimation of $\theta$ at the receiver is
\begin{equation}\label{sju}
\begin{aligned}
	\hat{\theta}_{t}=\hat{\theta}_{t-1}-\mu_{t}Y_t
	=\hat{\theta}_{t-1}-\omega\mu_{t}S_{t}-\mu_t(Y_{t-1}-\omega S_{t-1}),\\
\end{aligned}
\end{equation}
 Substituting (\ref{es1}) and (\ref{eded}) into (\ref{sju}), we have
\begin{equation}\label{est-t}
		\hat{\theta}_{t}
		=\theta+\omega\sum_{i=t+1}^{n}\mu_{i}S_i+\varepsilon_t.
\end{equation}
According to (\ref{est-t}),
at time $n$, $\hat{\theta}_{n}=\theta+\varepsilon_n$, which is the same as the classical SK scheme and $\varepsilon_n$ is a Gaussian variable with zero-mean and variance $\alpha_n=\frac{\sigma^2}{12\gamma P}\left(\frac{\gamma P}{\gamma P+\sigma^2}\right)^{n-1}$.
Following the error analysis of the SK-type scheme \cite{nm-tw,sk,dpc-fed}, we conclude that the achievable rate satisfies
$R \le  \frac{1}{2}\log \left(1+\frac{\gamma P}{\sigma^2}\right)$.

\textit{State Estimation:}
For $t=1$, $\hat{S}_1=0$ since $Y_1$ does not contain $S_1$.
For $t \ge 2$, since $Y_t$ only contains $S_t$, the receiver does the minimum mean square error (MMSE) estimation of $S_t$ based on the received signal $Y_t$, i.e.,
\begin{equation}
	\hat{S}_t=\frac{{\rm E}[S_tY_t]}{{\rm E}[Y_t^2]}Y_t\stackrel{(a)}=\frac{\sqrt{Q}(\sqrt{Q}+\sqrt{(1-\gamma)P})}{(\sqrt{Q}+\sqrt{(1-\gamma)P})^2+\gamma P+\sigma^2}Y_t,
\end{equation}
where (a) is because the SK-type codeword $G_t$  is a linear function about $(\eta_1,\eta_2,\ldots\eta_{t-1})$, and it is independent of $S_t$ and $\eta_t$.
Hence, the resulting MMSE distortion $D$ is
\begin{equation}
	\begin{aligned}
		D&=\lim_{n\to \infty}\frac{1}{n}\sum_{t=1}^{n}{\rm E}[(S_t-\hat{S}_t)^2]
		=
		\lim_{n\to \infty}\frac{{\rm E}[(S_1)^2]}{n}\\
		&+\lim_{n\to \infty} \frac{n-1}{n}Q\frac{\gamma P+\sigma^2}{(\sqrt{Q}+\sqrt{(1-\gamma)P})^2+\gamma P+\sigma^2}\\
		&=Q\frac{\gamma P+\sigma^2}{(\sqrt{Q}+\sqrt{(1-\gamma)P})^2+\gamma P+\sigma^2}.
	\end{aligned}
\end{equation}

By changing the power allocation coefficient $0\le \gamma \le 1$, the above scheme achieves the optimal R-D trade-off region $\mathscr{C}_{\rm dp}^{\rm  fb}$ given in (\ref{c1}).

\begin{figure}[h!]
	\centering
	\includegraphics[width=0.96\linewidth]{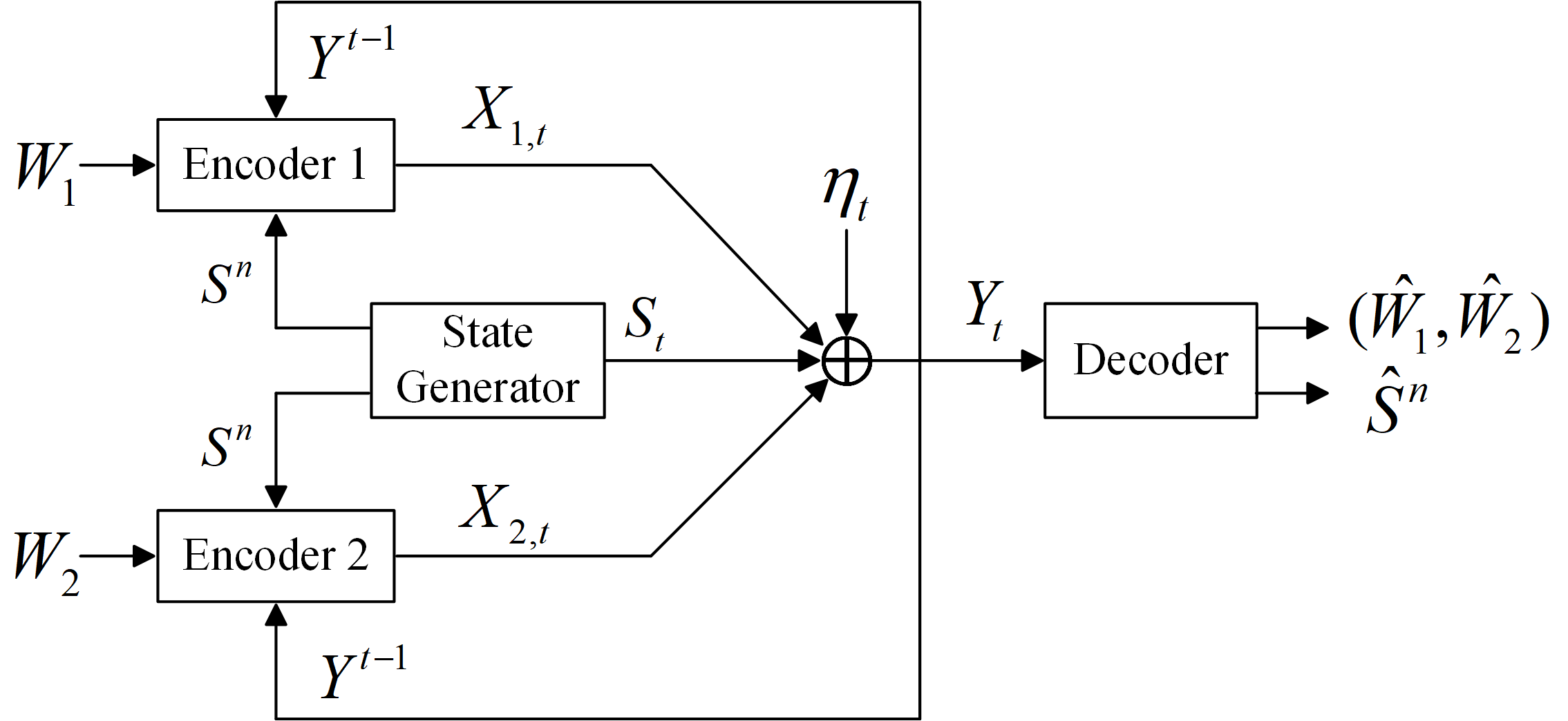}
	\caption{The DP-MAC with SE-R and feedback}
	\label{figmod-2}
\end{figure}

\section{The DP-MAC With SE-R and Feedback}\label{sec-3}

\subsection{Model Formulation}
The DP-MAC with SE-R and feedback is depicted in Fig. \ref{figmod-2}, and  the channel output at time instant $t$ is given by
\begin{equation}
	Y_t=X_{1,t}+X_{2,t}+S_t+\eta_t,  \, t=1,2,\ldots,n,
\end{equation}
where $X_{i,t}$ ($i \in \{1,2\}$) is the output signal of Encoder $i$, $\eta_t \sim \mathcal{N}(0,\sigma^2)$ is an i.i.d. channel noise, and $S_t \sim \mathcal{N}(0,Q)$ is an i.i.d. state interference. Both transmitters know the  state sequence $S^n=(S_1,S_2,\ldots,S_n)$ in a non-causal manner.
In addition, we assume that $\eta_t$, $S_t$ and the messages $W_1,W_2$  are independent of each other.

An $(n,R_1,R_2,P_1,P_2)$-code for the model of Fig. \ref{figmod-2} consists of:
\begin{enumerate}[1)]
	\item two messages $W_1$ and $W_2$, where $W_i$
	is uniformly distributed over the set $\mathcal{W}_i=\{1,2,\ldots,2^{nR_i}\}$ for $i \in \{1,2\}$;
	\item two encoders (Encoder $1$ and Encoder $2$), where Encoder $i$ ($i \in \{1,2\}$) with output $X_{i,t}=\varphi_{i,t}(W_i,S^n,Y^{t-1})$ satisfying the average power constraint $
	\frac{1}{n}\sum_{t=1}^{n}{\rm E}[X_{i,t}^2] \le P_i$;
	\item a decoder with outputs $(\hat{W}_1,\hat{W}_2)=\psi(Y^n)$ and $\hat{S}^n=\phi(Y^n)$, where $\psi$ and $\phi$ are the decoding and estimation functions, respectively.

\end{enumerate}
The average decoding error probability is defined as
\begin{equation}\label{pe-def2}
\begin{aligned}
	P_{e}^{(n)}=\frac{1}{2^{n(R_1+R_2)}}\sum_{w_1=1}^{2^{nR_1}}\sum_{w_2=1}^{2^{nR_2}}\Pr\{&\psi(Y^n)\neq (w_1,w_2)\\
	&|(w_1,w_2)\;\mbox{was sent}\},
\end{aligned}
\end{equation}
and the mean-squared state estimation error ${\rm E}d(S^n,\hat{S}^n)$ satisfies (\ref{ee-def}).

An R-D tuple $(R_1, R_2, D)$ is achievable if for any $\epsilon>0$,
there exists the above $(n,R_1,R_2,P_1,P_2)$-code such that  $\frac{1}{n}H(W_1) \ge R_1-\epsilon$, $\frac{1}{n}H(W_2)\ge R_2-\epsilon$, $P_e^{(n)}\le \epsilon$, and ${\rm E}d(S^n,\hat{S}^n)\le D+\epsilon$ as $n\to \infty$.
The optimal R-D trade-off region for the DP-MAC with  SE-R and feedback is the closure of the convex hull of all achievable $(R_1, R_2, D)$ tuples, denoted as $\mathscr{C}_{\rm dp,mac}^{\rm fb}$.

\subsection{Main Results}
The following Theorem $\ref{th01}$ provides a  complete characterization of $\mathscr{C}_{\rm dp,mac}^{\rm fb}$.

\begin{theorem}\label{th01}
$\mathscr{C}_{\rm dp,mac}^{\rm fb}=\bigcup_{0\le \rho \le 1}\mathcal{R}(\rho)$, and $\mathcal{R}(\rho)$ is given by
	\begin{equation}\label{c2}
		\begin{aligned}
		\mathcal{R}(\rho)&=\bigcup_{0\le \gamma \le 1}\bigcup_{0\le \beta \le 1}
		\Big
		\{(0\le R_1,0\le R_2,D):\\
		R_1 &\le \frac{1}{2}\log \Big(1+\frac{\gamma P_1(1-\rho^{2})}{\sigma^2}\Big),\\
		R_2 &\le \frac{1}{2}\log \Big(1+\frac{\beta P_2(1-\rho^{2})}{\sigma^2}\Big),\\
		R_1+R_2 &\le  \frac{1}{2}\log \Big(1+\frac{\gamma P_1+\beta P_2+2\sqrt{\gamma P_1 \beta P_2}\rho}{\sigma^2}\Big),\\
		D&\ge \frac{Q(\gamma P_1+ \beta P_2+\sigma^2+2\sqrt{\gamma P_1 \beta P_2}\rho)}{L(\gamma,\beta)+2\sqrt{\gamma P_1 \beta P_2}\rho}
		\Big\}.
		\end{aligned}
	\end{equation}
where
\begin{equation}\label{lde}
\begin{aligned}
	L(\gamma,\beta)&=P_1+P_2+Q+\sigma^2+2\sqrt{(1-\gamma)P_1Q}\\
	&+2\sqrt{(1-\beta)P_2Q}+2\sqrt{(1-\gamma)(1-\beta)P_1P_2}.
\end{aligned}
\end{equation}	
\end{theorem}
\begin{proof}
See the next subsection.
\end{proof}
\begin{remark}\label{rem1}
	For the  DP-MAC with SE-R and without feedback, the optimal R-D trade-off region was characterized in \cite{jsec-mac} and is given by
		\begin{equation}
		\begin{aligned}
			&\mathscr{C}_{\rm dp,mac}=\bigcup_{0\le \gamma \le 1}\bigcup_{0\le \beta \le 1}
			\Big
			\{(0\le R_1,0\le R_2,D):\\
			&R_1 \le \frac{1}{2}\log \Big(1+\frac{\gamma P_1}{\sigma^2}\Big),
			R_2 \le \frac{1}{2}\log \Big(1+\frac{\beta P_2}{\sigma^2}\Big),\\
			&R_1+R_2 \le  \frac{1}{2}\log \Big(1+\frac{\gamma P_1+\beta P_2}{\sigma^2}\Big),\\		
			&D\ge \frac{Q(\gamma P_1+ \beta P_2+\sigma^2)}{L(\gamma,\beta)}
			\Big\}.
		\end{aligned}
	\end{equation}
where  $L(\gamma,\beta)$ is given in (\ref{lde}).
Letting $\rho=0$ in (\ref{c2}), we conclude that $\mathscr{C}_{\rm dp,mac}=\mathcal{R}(0)$. Then by changing $0 \le  \rho \le 1$, $\mathscr{C}_{\rm dp,mac}^{\rm fb}=\bigcup_{0\le \rho \le 1}\mathcal{R}(\rho)$ achieves a larger region than $\mathscr{C}_{\rm dp,mac}$, which indicates that \textit{feedback enhances the  optimal R-D trade-off region of the  DP-MAC with SE-R and without feedback}.
\end{remark}

\subsection{Proof of Theorem \ref{th01}}\label{sec3-3}
\subsubsection{Proof of Achievability}
The main idea of the SK-type feedback scheme for the DP-MAC with SE-R is briefly illustrated below.
Encoder $1$ uses power $(1-\gamma)P_1$ ($0\le \gamma \le 1$) to encode $S^n$ as $V_1^n=\sqrt{\frac{(1-\gamma)P_1}{Q}}S^n$,  uses power $\gamma P_1$ to encode $W_1$ and the feedback $Y^{n-1}$ as $G_1^n$, and then sends $X_1^n=V_1^n+G_1^n$.
Encoder $2$ uses power $(1-\beta)P_2$ ($0\le \beta \le 1$) to encode $S^n$ as $V_2^n=\sqrt{\frac{(1-\beta)P_2}{Q}}S^n$,  uses power $\beta P_2$ to encode $W_2$ and the feedback $Y^{n-1}$ as $G_2^n$, and then sends $X_2^n=V_2^n+G_2^n$.
Define $\lambda=1 +\sqrt{\frac{(1-\gamma)P_1}{Q}}+\sqrt{\frac{(1-\beta)P_2}{Q}}$,
the received channel output  can be written as
\begin{equation}\label{outp2}
		Y^n=X_1^n+X_2^n+S^n+\eta^n
			=G_1^n+G_2^n+\lambda S^n +\eta^n.
\end{equation}
Then $G_1^n$ and $G_2^n$ are constructed by the SK-type feedback scheme \cite{dpmac-fed} for the DP-MAC with a new state $\lambda S^n$.
The codewords $G_1^n$ and $G_2^n$ are described below.

At time $1$, $G_{2,1}=0$ and
\begin{equation}
	G_{1,1}=\sqrt{12\gamma P_1}(\theta_1-O_1),
\end{equation}
where $\theta_1=-\frac{1}{2}+\frac{(2W_1-1)}{2\cdot2^{nR_1}}$ and $O_1=\frac{\lambda}{\sqrt{12\gamma P_1}}S_1-\lambda \sum_{i=3}^{n}\mu_{1,i}S_i$.
The receiver gets
$Y_1=\sqrt{12\gamma P_1}(\theta_1-O_1) + \lambda S_1 +\eta_1$,
and calculates
$\hat{\theta}_{1,1}=\frac{Y_1}{\sqrt{12\gamma P_1}}=\theta_1+\lambda \sum_{i=3}^{n}\mu_{1,i}S_i+\varepsilon_{1,1}$, where $\varepsilon_{1,1}=\frac{\eta_1}{\sqrt{12\gamma P_1}}$, and $\alpha_{1,1}=\frac{\sigma^2}{12\gamma P_1}$.

At time $2$, $G_{1,2}=0$ and
\begin{equation}
	G_{2,2}=\sqrt{12\beta P_2}(\theta_2-O_2),
\end{equation}
where $\theta_2=-\frac{1}{2}+\frac{(2W_2-1)}{2\cdot2^{nR_2}}$ and $O_2=\frac{\lambda}{\sqrt{12\beta P_2}}S_2-\lambda \sum_{i=3}^{n}\mu_{2,i}S_i$.
The receiver gets
$Y_2=\sqrt{12\beta P_2}(\theta_2-O_2)+ \lambda S_2 +\eta_2$,
and calculates
$\hat{\theta}_{2,2}=\frac{Y_2}{\sqrt{12\beta P_2}}=\theta_2+\lambda \sum_{i=3}^{n}\mu_{2,i}S_i+\varepsilon_{2,2}$, where $\varepsilon_{2,2}=\frac{\eta_2}{\sqrt{12\beta P_2}}$, and $\alpha_{2,2}=\frac{\sigma^2}{12\beta P_2}$.
Let $\varepsilon_{1,2}\triangleq\varepsilon_{1,1}$ and $\alpha_{1,2}\triangleq\alpha_{1,1}$.

At time $t$ ($t\in\{3,\ldots,n\}$),
\begin{equation}\label{g1g2}
G_{1,t}=\sqrt{\frac{\gamma P_1}{\alpha_{1,t-1}}}\varepsilon_{1,t-1}, G_{2,t}=\text{sgn}(\rho_{t-1})\sqrt{\frac{\beta P_2}{\alpha_{2,t-1}}}\varepsilon_{2,t-1},
\end{equation}
where  for $i\in\{1,2\}$, $\alpha_{i,t-1}\triangleq {\rm Var}[\varepsilon_{i,t-1}]$,
$\varepsilon_{i,t-1}=\varepsilon_{i,t-2}-\mu_{i,t-1}(Y_{t-1}-\lambda S_{t-1})$,
$\mu_{i,t-1}=\frac{{\rm E}[\varepsilon_{i,t-2}(Y_{t-1}-\lambda S_{t-1})]}
{{\rm E}[(Y_{t-1}-\lambda S_{t-1})^2]}$,
$\rho_{t-1}=\frac{{\rm E}[\varepsilon_{1,t-1}\varepsilon_{2,t-1}]}{\sqrt{\alpha_{1,t-1}\alpha_{2,t-1}}}$ is the correlation coefficient between $\varepsilon_{1,t-1}$ and $\varepsilon_{2,t-1}$, $\text{sgn}(\rho_{t-1})$ is equal to $1$ if $\rho_{t-1}\ge 0$ and $0$ otherwise.

The analysis of achievable rates is similar to the SK-type schemes \cite{dpmac-fed,sk-mac} for the multiple-access channels. Hence we omit some steps of the proof and conclude that the achievable rates of this scheme satisfy
\begin{equation}
\begin{aligned}
&R_1\le \frac{1}{2}\log \Big(1+\frac{\gamma P_1 (1-\rho^{*2})}{\sigma^2}\Big),\\
&R_2 \le \frac{1}{2}\log \Big(1+\frac{\beta P_2 (1-\rho^{*2})}{\sigma^2}\Big),\\
&R_1+R_2\le \frac{1}{2}\log \Big(1+\frac{\gamma P_1+ \beta P_2+2\sqrt{\gamma P_1 \beta P_2}\rho^*}{\sigma^2}\Big),
\end{aligned}
\end{equation}
where $\rho^*$ is the unique positive root in $(0,1)$ of
\begin{equation}
\begin{aligned}
&\sigma^2\left(\gamma P_1+ \beta P_2+2\sqrt{\gamma P_1 \beta P_2}\rho+\sigma^2\right)\\
&=\left(\beta P_2 (1-\rho^{*2})+\sigma^2\right)\left(\gamma P_1 (1-\rho^{*2})+\sigma^2\right).
\end{aligned}
\end{equation}

The analysis of state estimation is similar to that of Section \ref{sec2-2}.
In this scheme, $\hat{S}_1=0$ and $\hat{S}_2=0$.
For $t\ge 3$, $\hat{S}_t=\frac{{\rm E}[S_tY_t]}{{\rm E}[Y_t^2]}Y_t$.
Hence, from (\ref{ee-def}), (\ref{outp2}) and (\ref{g1g2}),
the resulting MMSE distortion $D$ is calculated as
\begin{equation}
	D=\frac{Q(\gamma P_1+ \beta P_2+\sigma^2+2\sqrt{\gamma P_1 \beta P_2}\rho^*)}{L(\gamma,\beta)+2\sqrt{\gamma P_1 \beta P_2}\rho^*},
\end{equation}
where $L(\gamma,\beta)$ is given in (\ref{lde}).


The above scheme achieves the region $\mathcal{R}(\rho^*)$.
Following \cite{dpmac-fed,sk-mac}, we conclude that the region $\mathcal{R}(\rho)$ for $\rho^*< \rho \le 1$ is strictly contained in $\mathcal{R}(\rho^*)$, and the region $\mathcal{R}(\rho)$ for $0\le  \rho <\rho^*$ is achieved by combining the random coding and the above SK-type scheme.
The achievability proof is completed.

\subsubsection{Proof of Converse}
The main idea of the converse is that due to the presence of feedback,
a new correlation coefficient is introduced between channel inputs $X_1$ and $X_2$ given $S$, in comparison to the converse proof of the DP-MAC with SE-R and without feedback \cite{jsec-mac}. The detailed proof is given below.

First, rate $R_1$ is bounded by
\begin{equation}\label{R1}
\begin{aligned}
	&nR_1=H(W_1) \stackrel{(a)}\le I(W_1;Y^n|W_2,S^n)+n\epsilon_{n}\\
	&\stackrel{(b)} \le \sum_{i=1}^{n} [h(Y_i|X_{2,i},S_i)-h(Y_i|X_{1,i},X_{2,i},S_i)]
+n\epsilon_{1n}
\end{aligned}	
\end{equation}
where (a) follows from Fano's inequality and $\epsilon_{n} \to 0$ as $n \to \infty$, and (b) follows from the chain rule and data processing inequality.
The upper bound of $R_2$ is similar.
Analogously, the sum rate $R_1+R_2$ is bounded by
\begin{equation}\label{RS}
\begin{aligned}
&n(R_1\!+\! R_2)=H(W_1,W_2)\le I(W_1,W_2;Y^n|S^n)\!+\! n\epsilon_{n}\\
&\le \sum_{i=1}^{n}[h(Y_i|S_i)-h(Y_i|X_{1,i},X_{2,i},S_i)]+n\epsilon_{n}.
\end{aligned}
\end{equation}

According to Lemma 1 in \cite{jsec-mac},  any communication scheme achieving a
distortion $D_n=\frac{1}{n}\sum_{t=1}^{n}{\rm E}[(S_t-\hat{S}_t)^2]$ over block length $n$ will have $\frac{1}{2}\log \frac{Q}{D_n} \le \frac{1}{n}I(S^n;Y^n)$.
Next, the sum $R_1+R_2+\frac{1}{2}\log \frac{Q}{D_n}$ is bounded by
{\setlength\abovedisplayskip{0.1cm}
	\setlength\belowdisplayskip{0.1cm}\begin{equation}\label{rs-d}
	\begin{aligned}
	R_1&+R_2 +\frac{1}{2}\log \frac{Q}{D_n}\le
	 \frac{1}{n}I(W_1,W_2;Y^n|S^n)\\
	 &+\frac{1}{n}I(S^n;Y^n)\!+\!\epsilon_{n}=\frac{1}{n}I(W_1,W_2,S^n;Y^n)+\! \epsilon_{n}\\
	 &\le \frac{1}{n}\sum_{i=1}^{n}[h(Y_i)-h(Y_i|X_{1,i},X_{2,i},S_i)] +\!\epsilon_{n}.
	\end{aligned}
\end{equation}}%
Assume that $X_{1,i}$, $X_{2,i}$ and $S_{i}$ have zero means, variances $P_1$, $P_2$ and $Q$. In addition, the covariance matrix of the vector $[X_{1,i},X_{2,i},S_i]^T$ is assumed to be
{\setlength\abovedisplayskip{0.1cm}
	\setlength\belowdisplayskip{0.1cm}\begin{equation}\label{ks}
	K_{X_{1,i},X_{2,i},S_i}=
	\begin{bmatrix}
		P_1	& a_i\sqrt{P_1P_2} & b_i\sqrt{P_1Q} \\
		a_i\sqrt{P_1P_2}	& P_2 & c_i\sqrt{P_2Q} \\
		b_i\sqrt{P_1Q}	& c_i\sqrt{P_2Q} & Q
	\end{bmatrix}.
\end{equation}}%
From (\ref{ks}), the correlation coefficient between $X_{1,i}$ and $X_{2,i}$ given $S_{i}$  is $\rho_i=\frac{a_i-b_ic_i}{\sqrt{1-b_i^2}\sqrt{1-c_i^2}}$
(here note that $\rho_i=0$ in the converse proof of \cite{jsec-mac}).
Define \begin{equation}\label{def-p}
	\gamma \triangleq \frac{1}{n} \sum_{i=1}^{n} (1-b_i^2), \beta \triangleq \frac{1}{n} \sum_{i=1}^{n} (1-c_i^2), \rho \triangleq \frac{1}{n} \sum_{i=1}^{n} \rho_i.
\end{equation}
By calculating the differential entropy in (\ref{R1})-(\ref{RS}) and applying Jensen's inequality together with (\ref{def-p}), the upper bounds of $R_1$, $R_2$ and $R_1+R_2$ given in (\ref{c2}) are obtained.
Analogously, form (\ref{rs-d}), we have
\begin{equation}\label{d-it}
	D_n \ge \frac{Q\sigma^2 \cdot 2^{2(R_1+R_2)}}{L(\gamma,\beta)+2\sqrt{\gamma P_1 \beta P_2}\rho},
\end{equation}
where $L(\gamma,\beta)$ is given in (\ref{lde}).
Substituting a fixed $R_1+R_2=\frac{1}{2}\log \Big(1+\frac{\gamma P_1+\beta P_2+2\sqrt{\gamma P_1 \beta P_2}\rho}{\sigma^2}\Big)$ into (\ref{d-it}), the bound of distortion in  (\ref{c2}) is obtained.  The detailed proof is shown in Appendix \ref{App1}.
The converse proof is completed.

\section{Discussion: A Noisy Observation Case}\label{sec-4}

We consider a noisy state observation  case for the DPC with SE-R and feedback, where the transmitter obtains a non-causal channel state corrupted by AWGN, i.e.,
\begin{equation}
		\tilde{S}^n=S^n+Z^n,
\end{equation}
where $Z^n=(Z_1,Z_2,\ldots,Z_n)$ and $Z_t \sim \mathcal{N}(0,\sigma_z^2)$ is i.i.d. Gaussian noise for $t\in \{1,2,\ldots,n\}$.
Then, the channel input at time instant $t$ satisfies
 $X_t=\varphi_t(W,\tilde{S}^n,Y^{t-1})$.

Since $S_t$ and $\tilde{S}_t$ are  jointly Gaussian distributed, $S_t$ can be re-written as
\begin{equation}
	S_t=\kappa \tilde{S}_t+ \tilde{Z}_t,
\end{equation}
where $\kappa=\frac{Q}{Q+\sigma_z^2}$, and $\tilde{Z}_t\sim \mathcal{N}(0,\kappa\sigma_z^2)$ is independent of $\tilde{S}$.
At time instant $t$, the equivalent channel output is written as
\begin{equation}\label{eqchan}
	Y_t=X_t+\kappa \tilde{S}_t+\tilde{Z}_t+\eta_t.
\end{equation}
Based on this equivalent model, the proposed scheme in Section \ref{sec2-2}
can be extended to the noisy observation case by viewing $\kappa \tilde{S}_t$ as a new state known at the transmitter, and $\tilde{Z}_t+\eta_t\sim \mathcal{N}(0,\kappa\sigma_z^2+\sigma^2)$ as a new channel noise.
The achievable R-D trade-off region of this extended scheme is
given by
{\setlength\abovedisplayskip{0cm}
	\setlength\belowdisplayskip{0cm}
\begin{equation}\label{c1-noisy}
	\begin{aligned}
		&\mathcal{R}_{\rm dp-fb}^{\rm nso}=\bigcup_{0 \le \gamma \le 1}
		\Bigg
		\{(R,D):
		0\le R\le \frac{1}{2}\log \Big(1 \!+\frac{\gamma P}{\kappa\sigma_z^2+\sigma^2}\Big),\\
		&D\ge Q\frac{\gamma P+\sigma^2+\kappa \sigma_z^2+(1-\kappa)\left(\sqrt{\kappa Q}+\sqrt{(1-\gamma)P}\right)^2}{\gamma P+\left(\sqrt{\kappa Q}+\sqrt{(1-\gamma)P}\right)^2+\sigma^2+\kappa \sigma_z^2}\Bigg\}.
	\end{aligned}
\end{equation}}%

Note that \cite{state-noisy} first studied the noisy observation case for the DPC with SE-R and without feedback, and derived the inner and outer bounds on the optimal R-D trad-off region. Very recently, this optimal region was
totally determined in \cite{xu}.
It is easy to verify that the converse proof of \cite{xu} remains valid in the present of feedback. This indicates that the optimal R-D trad-off region of the DPC with SE-R, noisy state observation and feedback is equal to that of same model without feedback.

\vspace{-1em}
\begin{figure}[h!]
	\centering
	\includegraphics[width=0.99\linewidth]{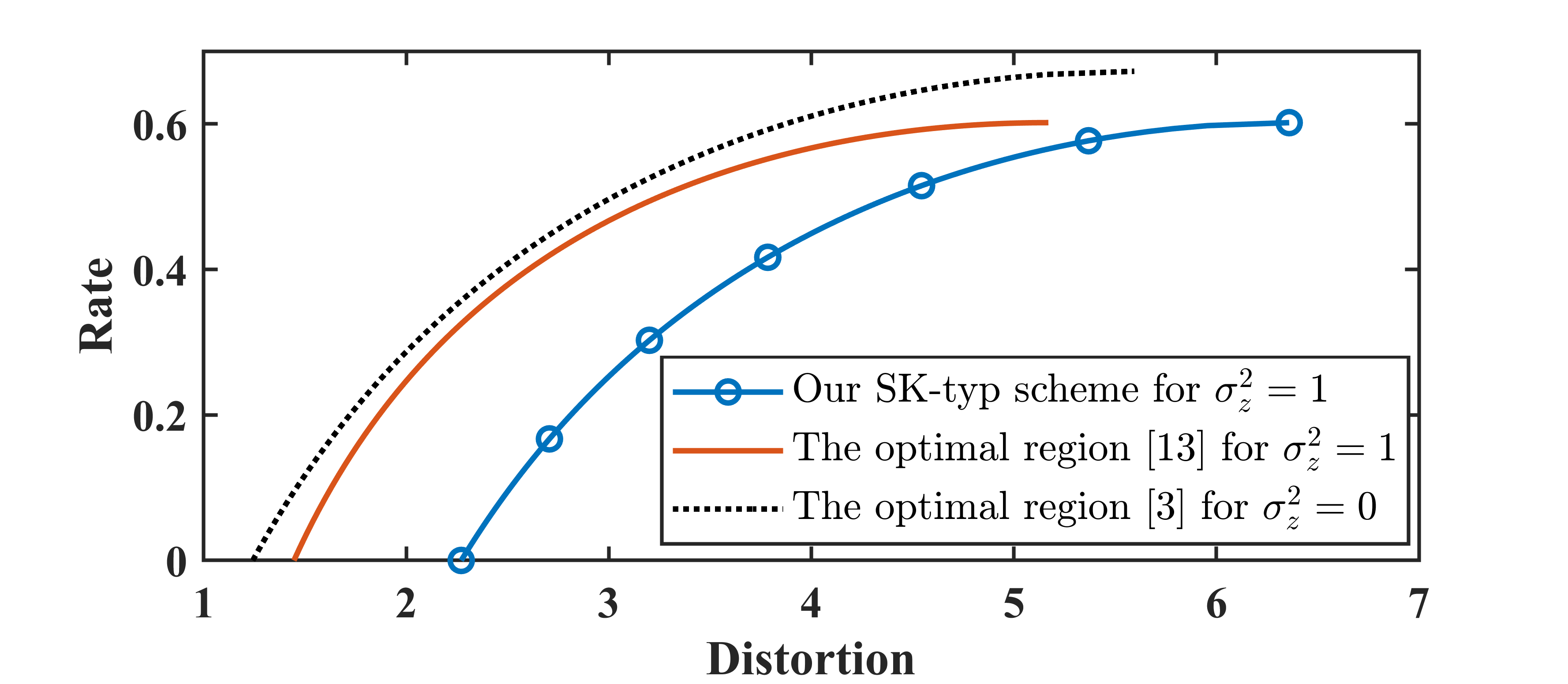}
	\caption{The comparison of $\mathscr{C}_{\rm dp}^{\rm  fb}$, the optimal region of \cite{xu}, and our SK-type scheme for $P=7.7$, $Q=10$, $\sigma^2=5$ and   $\sigma_z^2=1$.}
	\label{fignu}
\end{figure}

\vspace{-1em}
Fig. \ref{fignu} shows that our extended noisy observation scheme is not optimal. The reason is given below.
In the achievability scheme of \cite{xu},
the transmitter uses Gelfand-Pinsker coding \cite{G-P}  to transmit the message based on this equivalent channel (\ref{eqchan}), and the receiver estimates the state $S$ based on the Gelfand-Pinsker codeword $U$ and channel output $Y$.
While in our SK-type scheme, the receiver's estimation of $S$ is based only on $Y$, which leads to a reduction in the estimation performance.

Analogously, the proposed scheme of the DP-MAC in Section \ref{sec3-3} can be extended to the noisy state observation case.
However, this extension for the DP-MAC is not optimal which follows from the similar reason given above, hence we omit the detailed explanation here.

Our future work is to design capacity-achieving SK-type feedback schemes for dirty paper channels with SE-R and noisy state observation at the transmitter.

Finally, we would like to note that the assumption of noiseless feedback links, although not realistic in general, is widely assumed in the feedback communication literature. This assumption has important theoretical implications and may be a useful step towards investigating more realistic noisy feedback scenarios.
Based on the schemes proposed in this paper,
we are currently working on the design of SK-type schemes with rate-limited and noisy feedback.




\renewcommand{\theequation}{A\arabic{equation}}
\appendices
\section{The Converse Part of Theorem 1}\label{App1}
\setcounter{equation}{0}

First, rate $R_1$ is bounded by
\begin{equation}\label{R1-a}
	\begin{aligned}
		R_1=&\;\frac{1}{n}H(W_1)\stackrel{(a)}=\frac{1}{n}H(W_1|W_2)\\
		=&\;\frac{1}{n}\big[H(W_1|W_2)-H(W_1|Y^n,S^n,W_2)\\
		&\;\;\;\;\;\;+H(W_1|Y^n,S^n,W_2)\big]\\
		\stackrel{(b)}\le &\; \frac{1}{n}I(W_1;Y^n,S^n|W_2)+\epsilon_{n}\\
		=&\;\frac{1}{n}\big[I(W_1;Y^n|S^n,W_2)+I(W_1;S^n|W_2)\big]+\epsilon_{n}\\
		\stackrel{(c)}=&\;\frac{1}{n}\Bigg[\sum_{i=1}^{n}\big[h(Y_i|Y^{i-1},S^n,W_2)\\
		&\;\;\;\;\;\;-h(Y_i|Y^{i-1},S^n,W_2,W_1)\big]\Bigg]+\epsilon_{n}\\
		\stackrel{(d)}=&\;\frac{1}{n}\Bigg[\sum_{i=1}^{n}\big[h(Y_i|X_{2,i},Y^{i-1},S^n,W_2)\\
		&\;\;\;\;\;\;-h(Y_i|X_{1,i},X_{2,i},Y^{i-1},S^n,W_2,W_1)\big]\Bigg]+\epsilon_{n}\\
		\stackrel{(e)} \le &\; \frac{1}{n}\sum_{i=1}^{n} \big[h(Y_i|X_{2,i},S_i)-h(Y_i|X_{1,i},X_{2,i},S_i)\big]
		+\epsilon_{n},
	\end{aligned}	
\end{equation}
where
\begin{enumerate}[(a)]
	\item follows from the fact that $W_2$ is independent of $W_1$;
	\item follows from  conditioning reduces entropy and Fano's inequality, i.e.,
	\begin{equation}
		H(W_1|Y^n,S^n,W_2)\le H(W_1,W_2|Y^n,S^n) \le n\epsilon_{n}, \nonumber
	\end{equation}
	where $\epsilon_{n}\to 0$ as $n \to \infty$;
	\item follows from the fact that $W_1$, $S^n$ and $W_2$ are independent and chain rule;
	\item follows from the fact that $X_{1,i}=\varphi_{1,i}(W_1,Y^{i-1},S^{n})$ and $X_{2,i}=\varphi_{2,i}(W_2,Y^{i-1},S^{n})$;
	\item follows from conditioning reduces entropy and the fact that the channel output at time $i$	depends only on the state $S_i$ and the inputs $X_{1,i}$ and  $X_{2,i}$, i.e.,
\end{enumerate}
\begin{equation}
	\begin{aligned}
		H(Y_i|X_{2,i},Y^{i-1},S^n,W_2) &\le H(Y_i|X_{2,i},S_{i}),\nonumber\\	
		H(Y_i|X_{1,i},X_{2,i},Y^{i-1},S^n,W_2,W_1)&=H(Y_i|X_{1,i},X_{2,i},S_i).
	\end{aligned}
\end{equation}

Similarly,
\begin{equation}
	R_2\le \frac{1}{n}\sum_{i=1}^{n} \big[h(Y_i|X_{1,i},S_i)-h(Y_i|X_{1,i},X_{2,i},S_i)\big]
	+\epsilon_{n},
\end{equation}
and the sum rate $R_1+R_2$ is bounded by
\begin{equation}\label{RS-a}
	\begin{aligned}
		R_1+ R_2&=\frac{1}{n}H(W_1,W_2)
		\le \frac{1}{n}I(W_1,W_2;Y^n|S^n)+\epsilon_{n}\\
		&\le \frac{1}{n} \sum_{i=1}^{n}\big[h(Y_i|S_i)-h(Y_i|X_{1,i},X_{2,i},S_i)\big]+\epsilon_{n}.
	\end{aligned}
\end{equation}

According to Lemma 1 in \cite{jsec-mac},  any communication scheme achieving a
distortion $D_n=\frac{1}{n}\sum_{t=1}^{n}{\rm E}[(S_t-\hat{S}_t)^2]$ over block length $n$ will have
\begin{equation}
	\frac{1}{2}\log \frac{Q}{D_n} \le \frac{1}{n}I(S^n;Y^n).
\end{equation}
Next, the sum $R_1+R_2+\frac{1}{2}\log \frac{Q}{D_n}$ is bounded by
\begin{equation}\label{rs-d-a}
	\begin{aligned}
		R_1&+R_2+\frac{1}{2}\log \frac{Q}{D_n}\\
		&\le
		\frac{1}{n}I(W_1,W_2;Y^n|S^n)
		+\frac{1}{n}I(S^n;Y^n)+\epsilon_{n}\\
		&=\frac{1}{n}I(W_1,W_2,S^n;Y^n)+ \epsilon_{n}\\
		&\le \frac{1}{n}\sum_{i=1}^{n}\big[h(Y_i)-h(Y_i|X_{1,i},X_{2,i},S_i)\big] +\epsilon_{n}.
	\end{aligned}
\end{equation}

Assume that $X_{1,i}$, $X_{2,i}$ and $S_{i}$ have zero means, variances $P_1$, $P_2$ and $Q$. In addition, the covariance matrix of the vector $[X_{1,i},X_{2,i},S_i]^T$ is assumed to be
\begin{equation}\label{ks-a}
	K_{X_{1,i},X_{2,i},S_i}=
	\begin{bmatrix}
		P_1	& a_i\sqrt{P_1P_2} & b_i\sqrt{P_1Q} \\
		a_i\sqrt{P_1P_2}	& P_2 & c_i\sqrt{P_2Q} \\
		b_i\sqrt{P_1Q}	& c_i\sqrt{P_2Q} & Q
	\end{bmatrix},
\end{equation}
where $-1\le a_i \le 1$, $-1\le b_i \le 1$, and $-1\le c_i \le 1$.
From (\ref{ks-a}),   the covariance matrix of the error vector of
the MMSE estimate of ${\bf X}_i=[X_{1,i},X_{2,i}]^T$ given $S_i$ is
\begin{equation}\label{givs-k}
	\begin{aligned}
		K_{{\bf X}_i|S_i}&=K_{{\bf X}_i}-K_{{\bf X}_iS_i} K_{S_i}^{-1}K_{S_i{\bf X}_i}\\
		&=\begin{bmatrix}
			(1-b_i^2)P_1	& (a_i-b_ic_i)\sqrt{P_1P_2}  \\
			(a_i-b_ic_i)\sqrt{P_1P_2}	& (1-c_i^2)P_2 \\
		\end{bmatrix},
	\end{aligned}
\end{equation}
and the correlation coefficient between $X_{1,i}$ and $X_{2,i}$ given $S_i$ is denoted as
\begin{equation}\label{rhoi}
	\rho_i =\frac{a_i-b_ic_i}{\sqrt{1-b_i^2}\sqrt{1-c_i^2}}.
\end{equation}
Here note that $\rho_i=0$ in the converse proof of \cite{jsec-mac}.
Then we calculate the differential entropy in (\ref{R1-a})-(\ref{RS-a}) and (\ref{rs-d-a}).

According to $Y_i=X_{1,i}+X_{2,i}+S_i+\eta_i$, we have
\begin{equation}\label{wfs-1}
	h(Y_i|X_{1,i},X_{2,i},S_i)=h(\eta_i)\stackrel{(a)}=\frac{1}{2}\log (2 \pi e \sigma^2),
\end{equation}
where (a) follows from the fact that $\eta_i$ is the Gaussian noise with zero-mean and variance $\sigma^2$,
and
\begin{equation}
	\begin{aligned}
		h(Y_i)&=h(X_{1,i}+X_{2,i}+S_i+\eta_i)\\
		&\le \frac{1}{2}\log \big(2 \pi e \big(P_{1}+P_{2}+Q+2a_i\sqrt{P_{1}P_{2}}\\ &\;\;\;+2b_i\sqrt{P_{1}Q}+2c_i\sqrt{P_{2}Q}+\sigma^2\big)\big),
	\end{aligned}	
\end{equation}
where the last step follows from maximum differential entropy lemma.

From (\ref{givs-k}) and (\ref{rhoi}), the term $h(Y_i|S_i)$ can be  written as
\begin{equation}
	\begin{aligned}
		h&(Y_i|S_i)=h(X_{1,i}+X_{2,i}+\eta_i|S_i)\\
		&\le \frac{1}{2}\log \Big(2\pi e \Big((1-b_i^2)P_{1}+(1-c_i^2)P_{2}\\
		&\;\;\;\;+2(a_i-b_ic_i)\sqrt{P_1P_2}+\sigma^2\Big)	
		\Big)\\
		&= \frac{1}{2}\log \Big(2 \pi e \Big((1-b_i^2)P_{1}+(1-c_i^2)P_{2}\\
		&\;\;\;\;+2\sqrt{P_{1}P_{2}}\sqrt{1-b_i^2}\sqrt{1-c_i^2}\rho_i+\sigma^2\Big)\Big).
	\end{aligned}
\end{equation}

Analogously, we have
\begin{equation}
	\begin{aligned}
		h&(Y_i|X_{2,i},S_i)=h(X_{1,i}+\eta_i|X_{2,i},S_i)\\
		&\le \frac{1}{2}\log \left(2 \pi e \left({\rm E}[(X_{1,i}-\hat{X}_{1,i}(X_{2,i},S_i))^2]+\sigma^2\right)\right)\\
		&= \frac{1}{2}\log \left(2 \pi e \left((1-b_i^2)P_{1}(1-\rho_i^2)+\sigma^2\right)\right),
	\end{aligned}
\end{equation}
where $\hat{X}_{1,i}(X_{2,i},S_i)$ is the MMSE estimate of $X_{1,i}$ given ${\bf V}_i=[X_{2,i},S_i]^T$, and the corresponding  minimum mean squared error is
\begin{equation}
	\begin{aligned}
		&K_{X_{1,i}}-K_{X_i{\bf V}_i}K_{{\bf V}_i}^{-1}K_{{\bf V}_iX_{1,i}}\\
		&=P_{1}\left(1-\frac{a_i^2-2a_ib_ic_i+b_i^2}{1-c_i^2}\right)
		=(1-b_i^2)P_{1}(1-\rho_i^2).
	\end{aligned}
\end{equation}
Similarly,
\begin{equation}\label{wfs-5}
	h(Y_i|X_{1,i},S_i)\le \frac{1}{2}\log \left(2 \pi e \left((1-c_i^2)P_{2}(1-\rho_i^2)+\sigma^2\right)\right).
\end{equation}

Define \begin{equation}\label{def-p-a}
	\gamma \triangleq \frac{1}{n} \sum_{i=1}^{n} (1-b_i^2), \beta \triangleq \frac{1}{n} \sum_{i=1}^{n} (1-c_i^2), \rho \triangleq \frac{1}{n} \sum_{i=1}^{n} \rho_i.
\end{equation}
Since $-1\le a_i,b_i,c_i \le 1$, we conclude that $0\le \gamma \le 1$, $0\le \beta \le 1$ and $-1\le \rho \le 1$.

By substituting (\ref{wfs-1})-(\ref{wfs-5}) into  (\ref{R1-a})-(\ref{RS-a}), and applying Jensen's inequality together with (\ref{def-p-a}), we conclude that
\begin{equation}\label{r12-up}
	\begin{aligned}
		&R_1\le  \frac{1}{2}\log \left(1+\frac{\gamma P_1}{\sigma^2}(1-\rho^2)\right),\\
		&R_2\le  \frac{1}{2}\log \left(1+\frac{\beta P_2}{\sigma^2}(1-\rho^2)\right),\\
		&R_1+R_2 \le \frac{1}{2}\log \left(1+\frac{\gamma P_{1}+\beta P_{2}+2\sqrt{\gamma \beta P_{1}P_{2}}\rho}{\sigma^2}\right),
	\end{aligned}
\end{equation}
Analogously, from (\ref{rs-d-a}), we have
\begin{equation}\label{r1r2d}
	R_1+R_2 +\frac{1}{2}\log \frac{Q}{D_n}
	\le \frac{1}{2}\log \left(1+\frac{L(\gamma,\beta)+2\sqrt{\gamma P_1 \beta P_2}\rho}{\sigma^2}\right),
\end{equation}
where
\begin{equation}\label{lde-a}
	\begin{aligned}
		L(\gamma,\beta)&=P_1+P_2+Q+\sigma^2+2\sqrt{(1-\gamma)P_1Q}\\
		&+2\sqrt{(1-\beta)P_2Q}+2\sqrt{(1-\gamma)(1-\beta)P_1P_2}.
	\end{aligned}
\end{equation}	

Equation (\ref{r1r2d}) can be  written as
\begin{equation}\label{d-it-a}
	D_n \ge \frac{Q\sigma^2 \cdot 2^{2(R_1+R_2)}}{L(\gamma,\beta)+2\sqrt{\gamma P_1 \beta P_2}\rho}.
\end{equation}
Substituting a fixed $R_1+R_2=\frac{1}{2}\log \Big(1+\frac{\gamma P_1+\beta P_2+2\sqrt{\gamma P_1 \beta P_2}\rho}{\sigma^2}\Big)$ into (\ref{d-it-a}), the distortion is bounded by
\begin{equation}\label{dn-low}
	D_n \ge \frac{Q(\gamma P_1+ \beta P_2+\sigma^2+2\sqrt{\gamma P_1 \beta P_2}\rho)}{L(\gamma,\beta)+2\sqrt{\gamma P_1 \beta P_2}\rho}.
\end{equation}
Hence, for any $0\le \gamma \le 1$, $0\le \beta \le 1$ and $-1\le \rho \le 1$, the closure of the convex hull of all regions that satisfy (\ref{r12-up}) for rate and (\ref{dn-low}) for distortion is the outer bound of the optimal R-D trade-off region $\mathscr{C}_{\rm dp,mac}^{\rm fb}$ of the DP-MAC with  SE-R and feedback.
Finally, since the region that satisfies (\ref{r12-up}) and (\ref{dn-low}) for $-1\le \rho <0$ is contained in the region for $0\le \rho \le 1$, the converse proof is completed.


\begin{thebibliography}{1}

\bibitem{co}
M. H. M. Costa, ``Writing on dirty paper,'' \emph{IEEE Trans. Inf. Theory}, vol. 29, no. 3, pp. 439-441, May 1983.


\bibitem{jsec-first}
 A. Sutivong, M. Chiang, T. M. Cover, and Y.-H. Kim, ``Channel capacity and state estimation for state-dependent Gaussian channels,'' \emph{IEEE Trans. Inf. Theory}, vol. 51, no. 4, pp. 1486–1495, Apr. 2005.
	
\bibitem{jsec-fb}	
S. I. Bross and A. Lapidoth, ``The rate-and-state capacity with feedback,'' \emph{IEEE Trans. Inf. Theory}, vol. 64, no. 3, pp. 1893-1918, Mar. 2018.

\bibitem{app1}	
H. Permuter and T. Weissman, ``Source coding with a side information `vending machine',''  \emph{IEEE Trans. Inf. Theory}, vol. 57, no. 7, pp. 4530-4544, Jul. 2011.

\bibitem{app2}
B. Ahmadi and O. Simeone, ``Distributed and cascade lossy source coding with a side information `vending machine','' \emph{IEEE Trans. Inf. Theory}, vol. 59, no. 10, pp. 6807-6819, Oct. 2013.


	
\bibitem{sk} J. P. M. Schalkwijk and T. Kailath, ``A coding scheme
for additive noise channels with feedback. part I: No bandwidth constraint,'' \emph{IEEE Trans. Inf. Theory}, vol. 12, pp. 172-182, Apr. 1966.	

\bibitem{dpc-fed}
N. Elia and J. Liu, ``Writing on dirty paper with feedback,''
\emph{Commun. Inf. Syst.}, vol. 5, no. 4, pp. 401-422, May 2005.

\bibitem{nm-tw}
N. Merhav and T. Weissman,
``Coding for the feedback Gel'fand-Pinsker channel and the feedforward Wyner-Ziv source,''
\emph{IEEE Trans. Inf. Theory}, vol. 52, no. 9, pp. 4207-4211, Sep. 2006.


\bibitem{jsec-mac}
V. Ramachandran, S. R. B. Pillai, and V. M. Prabhakaran, ``Joint state estimation and communication over a state-dependent Gaussian multiple access channel,'' \emph{IEEE Trans. Commun.}, vol. 67, no. 10, pp. 6743-6752, Oct. 2019.

\bibitem{dpmac-fed}
A. Rosenzweig, ``The capacity of Gaussian multi-user channels with state and feedback,'' \emph{IEEE Trans. Inf. Theory}, vol. 53, no. 11, pp. 4349-4355, Nov. 2007.

\bibitem{sk-mac}
L. Ozarow,
``The capacity of the white Gaussian multiple access channel with feedback,''
\emph{IEEE Trans. Inf. Theory}, vol. 30, no. 4, pp. 623-629, Jul. 1984.


\bibitem{state-noisy}
C. Tian, B. Bandemer, and S. Shamai Shitz, ``Gaussian state amplification with noisy observations,'' \emph{IEEE Trans. Inf. Theory}, vol. 61, no. 9, pp. 4587-4597, Sep. 2015.

\bibitem{xu}
Y. Xu, T. Guo, D. Cao, and W. Xu, ``Capacity-Distortion tradeoff of noisy Gaussian state amplification,'' in \emph{Proc. 2023 IEEE International Symposium on Information Theory (ISIT)},  Jun. 2023, pp. 1836-1841.

\bibitem{G-P}
S. I. Gelfand and M. S. Pinsker, ``Coding for channel with random parameters,'' \emph{Probl. Contr. Inf. Theory}, vol. 9, no. 1, pp. 19-31, 1980.

%
\end{thebibliography}
\end{document}